\documentclass[runningheads,letterpaper,oribibl]{llncs}

\usepackage{amssymb}
\usepackage{amsmath}
\usepackage{amsfonts}  
\usepackage{xspace}
\usepackage{algorithm} 
\usepackage{algpseudocode}
\usepackage[usenames,dvipsnames]{color}

\newcommand{\RR}{\mathbb R}

\newcommand{\AAA}{\mathcal A}

\newcommand{\VVV}{\mathcal V}

\newcommand{\argmax}{\operatorname{argmax}}

\newcommand{\Ind}{\operatorname{Id}}

\newcommand{\update}{\operatorname{UPDATE}}

\newcommand{\bin}{\operatorname{bin}}
\newcommand{\updatee}{\operatorname{update}}
\newcommand{\final}{\operatorname{final}}
\newcommand{\setup}{\operatorname{setup}}
\newcommand{\game}{budgeted game\xspace}

\newcommand{\ignore}[1]{}

\newif\ifcomments
\commentstrue
\ifcomments
\newcommand{\smriti}[1]{{\color{blue}#1}}
\newcommand{\udi}[1]{{\color{red}#1}}
\newcommand{\anthony}[1]{{\color{green}#1}}
\else
\newcommand{\smriti}[1]{}
\newcommand{\udi}[1]{}
\newcommand{\anthony}[1]{}
\fi

\newcommand{\junk}[1]{}

\newif\ifshowqed
\showqedtrue

\begin{document}
\title{The Shapley Value in Knapsack Budgeted Games}
\author{Smriti Bhagat\inst{1} \and Anthony Kim\inst{2}\thanks{This work was done while the author was an intern at Technicolor Research Lab. Supported in part by an NSF Graduate Research Fellowship.} \and S. Muthukrishnan\inst{3} \and Udi Weinsberg\inst{1} }
\institute{
Technicolor Research\\
\email{\{smriti.bhagat,udi.weinsberg\}@technicolor.com}
\and
Department of Computer Science, Stanford University \\ 
\email{tonyekim@stanford.edu}
\and
Rutgers University\\
\email{muthu@cs.rutgers.edu}}

\maketitle

\begin{abstract}
We propose the study of computing the Shapley value for a new class of cooperative games that we call budgeted games, and investigate in particular knapsack budgeted games, a version modeled after the classical knapsack problem. In these games, the ``value'' of a set $S$ of agents is determined only by a critical subset $T\subseteq S$ of the agents and not the entirety of $S$ due to a budget constraint that limits how large $T$ can be. We show that the Shapley value can be computed in time faster than by the na\"ive exponential time algorithm when there are sufficiently many agents, and also provide an algorithm that approximates the Shapley value within an additive error. For a related budgeted game associated with a greedy heuristic, we show that the Shapley value can be computed in pseudo-polynomial time. Furthermore, we generalize our proof techniques and propose what we term algorithmic representation framework that captures a broad class of cooperative games with the property of efficient computation of the Shapley value. The main idea is that the problem of determining the efficient computation can be reduced to that of finding an alternative representation of the games and an associated algorithm for computing the underlying value function with small time and space complexities in the representation size. 
\end{abstract}

\section{Introduction}\label{sec:intro}
The Shapley value is a well-studied solution concept for fair distribution of profit among agents in cooperative game theory. Given a coalition of agents that collectively generate some profit, fair distribution is important to maintain a stable coalition such that no subgroup of agents has an incentive to unilaterally deviate and form its own coalition. While the Shapley value is not a stability concept, it uniquely satisfies a set of desirable properties for fair profit distribution based on individual contributions. It has been shown useful on a wide range of cooperative games and, more recently, applied beyond the game-theoretic setting in problems related to social networks~\cite{Narayanam,Michalak} and computer networks~\cite{Ric07Internet,Misra:2010}. 

Efficient --- (pseudo) polynomial time --- computation of the Shapley value has been studied for many classes of cooperative games. One such example is weighted voting games that model parliamentary voting where agents are parties, the weight of each party is the number of the same party representatives, and a coalition of parties is winning (has value 1) if its total weight is at least some quota, or losing (has value 0) otherwise. It was shown that computing the Shapley value in the weighted majority games, where the quota is half the total weight of all the agents, is \#P-complete~\cite{Deng} and NP-hard~\cite{Matsui2}. Note, however, that there is a pseudo-polynomial time algorithm using dynamic programming~\cite{Matsui}.

In another line of research, representation schemes for cooperative games have been proposed in ~\cite{Conitzer,Ieong,Ieong2,Aadithya}; if a given cooperative game has a small alternative representation in one of these schemes, then the Shapley value can be computed efficiently in time polynomial in the size of the alternative representation. For example, we can represent a given cooperative game as a collection of smaller cooperative games in multi-issue representation~\cite{Conitzer}, or in terms of logic rules in marginal contribution net representation~\cite{Ieong}.

We propose a new class of cooperative games that we call {\em budgeted games} and study the Shapley value computation in these games. In cooperative games, the value function $v(S)$ for a coalition $S$ is determined by all the agents in $S$, but may explicitly depend on a sub-coalition in some domains (e.g., \cite{Bachrach2,Bachrach3,Aziz}). We study value functions conditioned on a budget $B$ where $v(S)$ may be totally determined by a potentially strict subset $T\subset S$ of agents. That is, budget $B$ models a physical or budget constraint that may limit the actual value of a coalition to be less than simply the total aggregate value of all the individual contributions and, hence, the profit generation of a coalition is determined only by a sub-coalition of the agents. There are many examples we can readily formulate as budgeted games to model real-life scenarios:  

\begin{itemize}
\item {\em (Graph Problems)}
Consider a network of nodes that correspond to facilities and edges between them that correspond to communication links. This can be modeled as 
a graph $G$ with weights on nodes. For any subset  $S$ of nodes, $v_B(S)$,  the value created by set $S$ under budget $B$, may be the maximum weight of an independent set of  at most $B$ nodes.
\item {\em (Set Problems)}
Let each agent be a sales agent targeting a specific set of customers. Then $v_B(S)$ may be the maximum number of customers that can be targeted by a subset of size at most $B$ of sales agents from $S$. 
\item {\em (Packing Problems)}
Consider creating a task force from a pool $S$ of avaliable agents where each agent is associated with some value and cost. Then $v_B(S)$ may be the largest total aggregate value from a subset of the agents with total cost at most $B$.

\item {\em (Data Mining Problems\footnote{This is the Shapley value computation problem for what is commonly known as the Top-$k$ problem.})}
Let each agent represent a document with some quality measure with respect to a fixed search query. We may approximate the total value of an ordered list of documents $S$, ordered by the quality measure, by those that appear at the top of the list. Then, the corresponding $v_B(S)$ is the sum of the top $B$ quality scores of documents in $S$.
\end{itemize}

For the Shapley value to be useful in value division problems modeled as budgeted games, we cannot simply apply the formula for the Shapley value as it would lead to an exponential time algorithm. Hence, it is important to understand its computational complexity in these games, and we study the knapsack version (equivalently, Packing Problems) in this paper. As far as we know, the budgeted games have not been studied previously. A related class of games called bin-packing games~\cite{binpacking,binpacking2,binpacking3} has been studied for different solution concepts of core and $\epsilon$-core.\footnote{While items and bins separate and bins model linear constraints in knapsack budgeted games, both items and bins are treated as agents and the goal is to share profit among them in a fair way in bin-packing games.} 

\paragraph{\bf Our Contributions.}

First, we propose a new class of cooperative games, {\em budgeted games}, and investigate the computational complexity of the Shapley value in a particular version of budgeted games. Second, we generalize our proof techniques and propose a general framework, {\em algorithmic representation}, for cooperative games. We note that all our algorithms have running times with a polynomial dependence on the number of agents. More specifically, our contributions are as follows:

\begin{itemize}
\item We study the knapsack version of \game{s} and show that computing the Shapley value in these games is NP-hard. On the other hand, we show that the Shapley value can be computed in time faster than by the na\"ive exponential time algorithm when there are sufficiently many agents. 

\item We provide an additive approximation scheme for the Shapley value via rounding; our approach does not use the standard sampling and normal distribution techniques~\cite{Bachrach,Fatima} in estimating the Shapley value. 

\item We consider the value function obtained by a 2-approximation greedy algorithm for the classical knapsack problem and show that for this function, the Shapley value can be computed in pseudo-polynomial time. 

\item We provide generalizations and present the algorithmic representation framework that captures a broad class of cooperative games with the property of efficient computation of the Shapley value. This includes many known classes of cooperative games in~\cite{Deng,Matsui,Michalak} and those with concise representations using schemes in~\cite{Conitzer,Ieong,Aadithya}.
 
\end{itemize}


\section{Preliminaries}
We represent the profit distribution problem as a {\em cooperative game} $(N,v)$ where $N$ is the set of agents and $v:2^N \rightarrow \RR$ is the {\em characteristic function} that assigns a value to each subset of agents, with $v(\emptyset)=0$. We also call $v$ the {\em value function} and use both characteristic and value functions interchangeably. For a subset of agents $S \subseteq N$, we interpret $v(S)$ as the value that these agents can generate collectively; $v(N)$ is the total value that the whole group generates.

The {\em Shapley value}~\cite{Shapley} is a solution concept based on marginal contributions that divides the total value $v(N)$ into individual shares $\phi_1, \ldots, \phi_{|N|}$ satisfying an intuitive notion of fairness. For $i\in N$ and $S\subseteq N\setminus \{i\}$, we define agent $i$'s {\em marginal contribution to $S$} to be $v(S\cup \{i\}) - v(S)$. 
The Shapley value is the unique profit distribution solution that satisfies the following properties:
\begin{enumerate}
\item (Efficiency) $\sum_{i\in N} \phi_i(v) = v(N)$; 
\item (Symmetry) If $v(S\cup \{i\}) - v(S) = v(S\cup \{j\}) - v(S)$ for all $S \subseteq N\setminus \{i,j\}$, then $\phi_i(v) = \phi_j(v)$;
\item (Null Player) If $v(S\cup \{i\}) - v(S) = 0$ for all $S\subseteq N\setminus \{i\}$, then $\phi_i(v) = 0$;
\item (Linearity) For any two cooperative games $(N,v)$ and $(N,w)$ and their combined game $(N, v+w)$, $\phi_i(v) + \phi_i(w) = \phi_i(v+w)$ for all $i\in N$.
\end{enumerate}

The Shapley value for each agent $i$ is computed as
\begin{equation}\label{eq:shapley_set}
\phi_i(v) = \sum_{S \subseteq N\setminus \{i\}} \frac{|S|! (|N|-|S|-1)!}{|N|!} (v(S\cup \{i\}) - v(S)).
\end{equation}

Note the Shapley value is a weighted average of agent $i$'s marginal contributions. Equivalently, it can also be computed as $\phi_i(v) = \frac{1}{|N|!} \sum_{\pi \in \Pi} v(P_\pi^i \cup \{i\}) - v(P_\pi^i)$, where $\Pi$ is the set of all $|N|!$ permutations of the agents and $P_\pi^i$ is the set of agents preceeding agent $i$ in the order represented by permutation $\pi$. 

There are two sources of computational complexity in the Shapley value: an exponential number of terms in the summation and individual evaluations of the characteristic function $v$. Directly applying the above equations leads to a na\"ive algorithm with running time at least exponential in the number of agents, $\Omega(2^{|N|})$; furthermore, each individual evaluation of $v$ can be expensive.


\section{Knapsack Budgeted Games}

A {\em knapsack \game} $(N,v)$ is a cooperative game with the alternative representation given by a nonnegative integer tuple $(\{(l_1, w_1), \ldots, (l_{|N|}, w_{|N|})\}, l_{\bin})$ such that $v(S) = \max_{S'\subseteq S: l(S')\leq l_{\bin}} w(S')$ for all $S\subseteq N$, where $l(S') = \sum_{k\in S'} l_k$ and $w(S') = \sum_{k\in S'} w_k$. Each agent $i$ is described by $(l_i, w_i)$ where $l_i$ and $w_i$ are the agent's length and weight, respectively. The variable $l_{\bin}$ is the bin size that restricts which set of agents can directly determine the value function $v$. For a set of agents $S$, the value $v(S)$ is determined by solving an optimization problem where the total value of selected agents, possibly a strict subset of $S$, is optimized subject to a budget constraint; the other unselected agents do not contribute explicitly. Note the similarities with the classical knapsack problem in which the objective is to find the maximum total value of items that can be packed into a fixed size bin. 

Knapsack budgeted games are useful when the characteristic function $v$ of a cooperative game can be modeled as the objective value of an optimization problem subject to linear constraints. In this paper, we only consider the games with a single linear constraint, but our results extend to knapsack budgeted games with multiple linear constraints. In a knapsack budgeted game with multiple budget constraints, each agent is associated with a length vector $\vec{l} = (l^1, \ldots, l^d)$ and a weight and there is a budget constraint on each coordinate, i.e., $l^1_{\bin}, \ldots, l^d_{\bin}$, assuming $d$ budget constraints. 

For an application, we can use knapsack budgeted games and the Shapley value to model value division in a sport team. We would like to give out bonuses proportional to the Shapley value solution. Assume each player $i$ is associated with a skill level $w_i$ and, in a game of the sport, at most $B$ players from each team can play. We model the value of the team as the total aggregate skill level of its best $B$ players, since they usually start and play the majority of the games. Then, this is a knapsack budgeted game with skill levels as weights, unit lengths, and $l_{\bin} = B$. Note the Shapley value of a player not in the top $B$ may be positive. Since he is still contributing to the team as a reserve player and might be one of the top $B$ players in a subset of the team, say available players in an event of injury, he should be compensated accordingly.

In the following sections, we assume that the knapsack budgeted game $(N,v)$ has the representation $(\{(l_1, w_1), \ldots, (l_{|N|}, w_{|N|})\}, l_{\bin})$. We define $w_{\max} = \lceil l_{\bin} \cdot \max_i w_i/l_i \rceil$, which is an upper bound on the value $v(N)$. We use shorthand notations $l(S) = \sum_{k\in S} l_k$ and $w(S) = \sum_{k\in S} w_k$ for any subset $S$. The set of agents are ordered and labeled with $1, \ldots, |N|$. For a set of agents $X$ and two integers $a$ and $b$, we use $X_{a,b}$ to denote the subset $\{i\in X: a\leq i \leq b\}$. To avoid degenerate cases, we further assume $0<l_i \leq l_{\bin}$ for all $i$. We use the indicator function $\Ind$ that equals to 1 if all the input conditions hold, or 0 otherwise.

\section{The Shapley Value in Knapsack Budgeted Games}\label{sec:shapley}
We present a hardness result, an algorithm for computing the Shapley value exactly, and a deterministic approximation scheme that approximates within an additive error. 
 
\subsection{Exact Computation}

By the NP-completeness of the classical knapsack problem and the efficiency property of the Shapley value, it follows that (see Appendix~\ref{sec:appendixA} for details):
\begin{theorem}\label{thm:nphard}
The problem of computing the Shapley value in the knapsack budgeted games is NP-hard. 
\end{theorem}

While a polynomial time algorithm for computing the Shapley value may or may not exist, the na\"ive exponential time algorithm is too slow when $|N|$ is large. When $|N|$ is sufficiently large, especially when $|N| \gg l_{\bin}$, we show that a faster algorithm exists:

\begin{theorem}\label{thm:exact}
In the knapsack budgeted games, the Shapley value can be computed in time $O(l_{\bin}(w_{\max}+1)^{l_{\bin}+1} |N|^2)$ for each agent. 
\end{theorem}

To prove Theorem~\ref{thm:exact}, we associate each subset $S\subseteq N$ with a vector from a finite-sized vector space that completely determines an agent's marginal contribution to $S$. If the cardinality of the vector space is small and the partitions of the $2^{|N|}$ subsets corresponding to the vectors can be found efficiently, we can evaluate $v$ once for each vector instead of once for each subset, reducing the overall computation time. Note that the well-known dynamic programming algorithm, call it $\AAA$, for the classical knapsack problem can be used to compute $v$; for a given $S$, the algorithm iteratively updates an integer array of length $l_{\bin}+1$ holding the optimal values for the sub-problems with smaller bin sizes and returns a final value determined by the array at termination.\footnote{Assume the agents in $S$ are labeled $1, \ldots, |S|$ for simplicity. For $1\leq j\leq |S|$, we define $c(j,b) = \max_{S'\subseteq S_{1,j}: l(S') \leq b} w(S')$. It has the recurrence relation $c(j,b) = \max\{c(j-1,b), c(j-1, b-l_j) + w_j\}$. We compute $c(j,b)$'s and, hence, $v(S) = c(|S|,l_{\bin})$ in $O(|S| l_{\bin})$ time. \label{fn:classicaldp}} We associate with each subset $S$ the final state of the array when $\AAA$ runs on $S$ and determine the cardinalities of resulting partitions using a dynamic program, different but related to $\AAA$; the dynamic program counts the number of optimal solutions to the sub-problems grouped by objective value while $\AAA$ simply computes the optimal solutions to the sub-problems. 

We use the following lemma to prove Theorem~\ref{thm:exact}; it shows that if the set of possible marginal contribution values for agent $i$ is small, then we can reduce the number of evaluations of $v$ by grouping subsets of $N\setminus \{i\}$ by marginal contribution value and evaluating $v$ once for each group (see Appendix~\ref{sec:appendixA} for a proof).

\begin{lemma}\label{lemma:char2}
Assume there exist positive integers $p_i$ and partition functions $P_i:2^{N \setminus \{i\}} \rightarrow \{1, \ldots, p_i\}$, for $i=1, \ldots, |N|$, such that if $P_i(S) = P_i(S')$ for two different $S, S' \subseteq N \setminus \{i\}$, then $v(S \cup \{i\}) - v(S) = v(S' \cup \{i\}) - v(S')$. Let $m_i(p)$ be agent $i$'s marginal contribution to $S$ for all $S$ satisfying $P_i(S) = p$, and $c(i,s,p) = \#\{S\subseteq N\setminus\{i\}: |S|=s, P_i(S) = p\}$ for $i\in N$, $0\leq s \leq |N|-1$, and $1 \leq p \leq p_i$. Then, the Shapley value for agent $i$ can be computed as  
\[
\textstyle \phi_i = \textstyle \sum_{p=1}^{p_i} \sum_{s=0}^{|N|-1} c(i,s,p) \frac{s!(|N|-s-1)!}{|N|!} m_i(p)
\]
in time $O(p_{\max}(t+q)|N|)$, where $p_{\max} = \max_i p_i$, $t$ is an upper bound on the computation time of the coefficients $c$, and $q$ is the evaluation time of $v$.
\end{lemma}

We now prove Theorem~\ref{thm:exact} by applying Lemma~\ref{lemma:char2}:

\begin{proof}(of Theorem~\ref{thm:exact})
We compute the Shapley value for some fixed agent $i$. We define $V_{A,b} = \max_{S'\subseteq A: l(S') \leq b} w(S')$, for $A \subseteq N$ and $0 \leq b \leq l_{\bin}$, and vector $\mathbf{V}_{S} = (V_{S,0}, \ldots, V_{S,l_{\bin}})$, for subsets $S \subseteq N$. Let $\VVV$ be the finite vector space $\{0, \ldots, w_{\max}\}^{l_{\bin}+1}$ that contains vectors $\mathbf{V}_S$. We use the 0-based index to indicate coordinates of a vector in $\VVV$; so, $v(S) = V_{S, l_{\bin}} = \mathbf{V}_S(l_{\bin})$ for all $S$. Given $\mathbf{V}_S$, agent $i$'s marginal contribution to $S$ can be computed in constant time as $v(S\cup \{i\}) - v(S) = \max\{\mathbf{V}_S(l_{\bin} - l_i) + w_i - \mathbf{V}_S(l_{\bin}), 0 \}$. Let this expression be defined more generally as $m_i(\mathbf{v}) = \max\{\mathbf{v}(l_{\bin} - l_i) + w_i - \mathbf{v}(l_{\bin}), 0 \}$ for $\mathbf{v}\in\VVV$. 

We partition $2^{N\setminus \{i\}}$ by the pair $(|S|, \mathbf{V}_S)$ so that for each possible $(s, \mathbf{v})$ pair, all subsets $S$ satisfying $|S|=s$ and $\mathbf{V}_S = \mathbf{v}$ are grouped together. Clearly, the marginal contribution of agent $i$ is the same within each partition. To compute the cardinality of each partition, we use dynamic programming. Let $N' = N\setminus\{i\}$, ordered and relabeled $1,\ldots, |N|-1$. For $0\leq j \leq |N|-1$, $0 \leq s \leq j$, and $\mathbf{v} \in \VVV$, we define $\hat{c}(j,s,\mathbf{v}) = \# \left\{S\subseteq N'_{1,j}: |S|=s, \mathbf{V}_S = \mathbf{v}\right\}$. Note $\hat{c}$ has the recurrence relation 
\[
\hat{c}(j,s,\mathbf{v}) = \textstyle \hat{c}(j-1, s, \mathbf{v}) + \sum_{\mathbf{u}: \update(\mathbf{u}, l'_j, w'_j)= \mathbf{v}} \hat{c}(j-1,s-1,\mathbf{u}),
\]
with the base case $\hat{c}(0, 0, \mathbf{0}) = 1$, where $l'_j$ and $w'_j$ correspond to the $j$-th agent in order in $N'$ and $\update$ is an $O(l_{\bin})$ algorithm that updates $\mathbf{u}$ with the additional agent: 1) Initialize $\mathbf{v} = \mathbf{u}$; 2) For $j=l'_j, \ldots, l_{\bin}$, $\mathbf{v}(j) = \max\{\mathbf{v}(j), \mathbf{u}(j-l_j')+w_j'\}$; and 3) Return $\mathbf{v}$.

%

Using the recurrence relation, we compute $\hat{c}(j,s, \mathbf{v})$ for all $j$, $s$, and $\mathbf{v}$ in time $O(l_{\bin}(w_{\max}+1)^{l_{\bin}+1} |N|^2)$. By Lemma~\ref{lemma:char2}, 
\[
\phi_i = \textstyle \sum_{\mathbf{v} \in \VVV} \sum_{s=0}^{|N|-1} \hat{c}(|N|-1, s,\mathbf{v}) \frac{s!(|N|-s-1)!}{|N|!} m_i(\mathbf{v}),
\]
and the Shapley value can be calculated in time $O((w_{\max}+1)^{l_{\bin}+1} |N|)$ using the precomputed values of $\hat{c}$. The computation of $\hat{c}$ dominates the application of the Shapley value equation, and the overall running time is $O(l_{\bin}(w_{\max}+1)^{l_{\bin}+1} |N|^2)$ per agent.
\ifshowqed
\hfill \qed
\fi 
\end{proof}

\subsection{Additive Approximation}\label{subsec:approx}
Similar to the fully polynomial time approximation scheme for the classical knapsack problem (see \cite{vazirani}), we show an approximation scheme for the Shapley value by rounding down the weights $w_i$'s and computing the Shapley value of the cooperative game $(N, v')$ where $v'$ is an approximation of $v$. Our technique of computing the Shapley value by approximating the characteristic function $v$ is deterministic and does not require concentration inequalities like the standard statistical methods of sampling and normal distribution techniques in \cite{Bachrach,Fatima}.  

The following lemma formalizes how an approximation of the characteristic function $v$ leads to an additive error in the Shapley value computation (see Appendix~\ref{sec:appendixA} for a proof):


\begin{lemma}\label{lemma:approx}
If $v'$ is an $\alpha$-additive approximation of $v$, i.e., $v'(S) \leq v(S) \leq v'(S) + \alpha$ for all $S\subseteq N$, then the Shapley value $\phi'_i$ computed with respect to $v'$ is within an $\alpha$-additive error of the Shapley value $\phi_i$ computed with respect to $v$, for all $i$.
\end{lemma}

When $w_{\max}$ is sufficiently larger than $l_{\bin}$, the approximation scheme's running time is faster than that of the exact algorithm of Theorem~\ref{thm:exact}:

\begin{theorem}\label{thm:approx}
In the knapsack budgeted games, the Shapley value can be computed within an $\epsilon w_{\max}$-additive error in $O(({l_{\bin}}^2 / \epsilon + 1)^{l_{\bin}+1} |N|^2)$ for each agent, where $\epsilon>0$.\footnote{For agent $i$, its Shapley value $\phi_i$ is clearly in $[0, w_{\max}]$. Using the approximation scheme, we can compute $\phi_i$ within $\frac{1}{7} w_{\max}$ for instance. As long as $\epsilon > {l_{\bin}}^2/w_{\max}$, the approximation scheme has a faster running time than the exact algorithm in Theorem~\ref{thm:exact}; this observation about $\epsilon$ is also true for the fully polynomial time approximation scheme for the classical knapsack problem (see \cite{vazirani}).}
\end{theorem}
\begin{proof}
We construct an approximate characteristic function $v'$ of $v$ as follows. Let $\epsilon >0$ and $k = \epsilon w_{\max} / l_{\bin}$. Note that when ${l_{\bin}}^2/\epsilon < w_{\max}$, $k>1$. For each agent $i$, let the rounded weight $w'_i$ be $\lfloor\frac{w_i}{k}\rfloor$. The lengths do not change. To compute $v'(S)$, we compute the optimal set $S' \subseteq S$, using dynamic programming, with respect to the rounded weights  $w'_1, \ldots, w'_{|N|}$ and let $v'(S) = k \sum_{i \in S'} w'_i$. In other words, $v'(S) = k \cdot \max_{S'\subseteq S: l(S') \leq l_{\bin}} w'(S)$ for all $S\subseteq N$, where we use the shorthand notation $w'(S) = \sum_{k \in S} w'_k$. 

We show $v(S) \geq v'(S) \geq v(S) - \epsilon w_{\max}$, for all $S\subseteq N$. Let $S$ be a subset and $T_O, T'\subseteq S$ be the optimal subsets using original and rounded weights, respectively, such that $v(S) = w(T_O)$ and $v'(S) = k\cdot w'(T')$. Note that both optimal sets have cardinality at most $l_{\bin}$. Because of rounding down, $w_i - k w'_i \leq k$ and $\sum_{j\in T_O} w_j - k \sum_{j\in T_O} w'_j \leq k l_{\bin}$. Since $T'$ is optimal with respect to the rounded weights, $\sum_{j\in T'} w'_j \geq \sum_{j\in T_O} w'_j$. Then, $v'(S) = k \sum_{j \in T'} w'_j \geq k \sum_{j \in T_O} w'_j \geq \sum_{j\in T_O} w_j - k l_{\bin} = v(S) - \epsilon w_{\max}$. Since $w_i \geq k w'_i$ for all $i$, $v(S) = w(T_O) \geq w(T') \geq k w'(T') = v'(S)$. Hence, $v'$ is an $\epsilon w_{\max}$-additive approximation of $v$. Then, the Shapley value computed with respect to $v'$ is within $\epsilon w_{\max}$ of the original Shapley value by Lemma~\ref{lemma:approx}. 

We now compute the Shapley value with respect to $v'$. For $A \subseteq N$, $0 \leq b \leq l_{\bin}$, we define $V'_{A,b} = \max_{S'\subseteq A: l(S') \leq b} w'(S')$. For a subset $S \subseteq N$, we define vector $\mathbf{V}'_{S} = (V'_{S,0}, \ldots, V'_{S,l_{\bin}})$. Note that $w'_i = \lfloor\frac{w_i}{k} \rfloor \leq \lfloor\frac{w_{\max}}{k} \rfloor = \lfloor \frac{l_{\bin}}{\epsilon} \rfloor$. Then, we can upper bound $w'(S) \leq l_{\bin} \lfloor \frac{l_{\bin}}{\epsilon} \rfloor$, for all $S$. Let $\VVV' = \{0, \ldots, l_{\bin}\lfloor \frac{l_{\bin}}{\epsilon} \rfloor\}^{l_{\bin}+1}$ that vectors $\mathbf{V}'_S$ are contained in. Note $v'(S) = k \cdot \mathbf{V}'_S(l_{\bin})$ for all $S$. From vector $\mathbf{V}'_S$, we can compute agent $i$'s marginal contribution to $S$ with respect to $v'$ in constant time: $v'(S\cup \{i\}) - v'(S) = k\cdot \max\{\mathbf{V}'_S(l_{\bin} - l_i) + w'_i - \mathbf{V}'_S(l_{\bin}), 0 \}$.

From here, we follow the proof of Theorem~\ref{thm:exact}. We compute the analogue of $\hat{c}$ in $O(({l_{\bin}}^2 / \epsilon + 1)^{l_{\bin}+1} |N|^2)$, and this is the dominating term in the Shapley value computation with respect to $v'$.
\ifshowqed
\hfill \qed
\fi 
\end{proof}

\section{Greedy Knapsack Budgeted Games}\label{sec:greedy}
Motivated by the approximation scheme in Theorem~\ref{thm:approx}, we investigate {\em greedy knapsack budgeted games}, a variant of knapsack budgeted games, and show the Shapley value in these games can be computed in pseudo-polynomial time. A greedy knapsack budgeted game has the same representation as the knapsack budgeted games, but its characteristic function is computed by a 2-approximation heuristic for the classical knapsack problem. We defer proofs to Appendix~\ref{sec:appendixB}.

\begin{algorithm}[t]
\caption{Greedy Heuristic $\AAA'(S,l_{\bin})$}\label{alg:greedy}
\begin{algorithmic}[1]
\State Let $a = \argmax_{k\in S}w_k$.
\State Select agents in $S$ in decreasing order of $\frac{w_i}{l_i}$ and stop when the next agent does not fit into the bin of size $l_{\bin}$; let $S'$ be the selected agents.
\State Return $S'$ if $w(S')\geq w_a$, or $\{a\}$ otherwise.
\end{algorithmic}
\end{algorithm}

\begin{theorem}\label{thm:greedy}
In the greedy knapsack budgeted games $(N,v)$ with $v(S) = \AAA'(S, l_{\bin})$ for all $S$, the Shapley value can be computed in $O({l_{\bin}}^5 {w_{\max}}^5 |N|^{8})$ for each agent, where the greedy heuristic $\AAA'(S, l_{\bin})$ is computed as in Algorithm~\ref{alg:greedy}.
\end{theorem}

While motivated by knapsack budgeted games, we use a different proof technique using the following lemma to prove Theorem~\ref{thm:greedy}. It generalizes the observation that in the simple cooperative game $(N, v)$ where the agents have weights $w_1, \ldots, w_{|N|}$ and the characteristic function $v$ is additive, i.e., $v(S) = \sum_{k\in S} w_k$, the Shapley value $\phi_i$ is exactly $w_i$ for all $i$. 

\begin{lemma}\label{lemma:char1}
Assume that the cooperative game $(N,v)$ has a representation $(M, w, A)$ where $M$ is a set, $w: M \rightarrow \RR$ is a weight function, and $A:2^N \rightarrow 2^M$ is a mapping such that $v(S) = \sum_{e \in A(S)} w(e)$, $\forall S \subseteq N$. Let $c_+(i,s,e) = \#\{S \subseteq N\setminus\{i\}: |S|=s, e \in A(S\cup\{i\})\}$ and $c_-(i,s,e) = \#\{S \subseteq N\setminus \{i\}: |S|=s,e \in A(S)\}$, for $i \in N$, $e \in M$, and $0 \leq s \leq |N|-1$. Then, the Shapley value for agent $i$ can be computed as 
\[
\phi_i = \textstyle \sum_{e\in M} \sum_{s=0}^{|N|-1} (c_+(i,s,e) - c_-(i,s,e)) \frac{s!(|N|-s-1)!}{|N|!} w(e).
\]
in time $O(t|M||N|)$ where $t$ is an upper bound on the computation time of the coefficients $c_+$ and $c_-$.
\end{lemma}

\section{Generalizations}\label{sec:general}

We present generalizations of our proof techniques and propose an unifying framework that captures a broad class of cooperative games in which computing the Shapley value is tractable, including many known classes of cooperative games in~\cite{Deng,Matsui,Michalak} and those with concise representations using schemes in~\cite{Conitzer,Ieong,Aadithya}. The main idea is that the problem of computing the Shapley value reduces to that of finding an efficient algorithm for the cooperative game's characteristic function. More precisely, if a cooperative game $(N,v)$ is described in terms of an alternative representation $I$ and an algorithm $A$ with low time and space complexities that computes $v$, formalized in terms of {\em decomposition}, then we can compute the Shapley value efficiently. To illustrate the generalizations' applicability, we use them to give examples of cooperative games in which the Shapley value can be computed efficiently. 

For each generalization, we consider two cases: the order-agnostic case in which $A$ processes agents in an arbitrary order, and the order-specific case in which $A$ processes in a specific order, like the greedy heuristic in Theorem~\ref{thm:greedy}. 

%

\begin{algorithm}[t]
\caption{Computing $A(I,S)$ with a decomposition $(A_{\setup}, A_{\updatee}, A_{\final})$} \label{alg:decomp} 
\begin{algorithmic}[1]	
\State $A_{\setup}(I)$ outputs $I'$, $\mathbf{x}$
\For{$i\in S$} 
\State $\mathbf{x} = A_{\updatee}(I',i,\mathbf{x})$ 
\EndFor
\State Return $A_{\final}(I', \mathbf{x})$
\end{algorithmic}
\end{algorithm}

\begin{definition}\label{defn:decomp}
Assume a cooperative game $(N,v)$ has an alternative representation $I$ and a deterministic algorithm $A$ such that $v(S) = A(I,S)$ for all $S\subseteq N$. Algorithm $A$ has a {\em decomposition $(A_{\setup}, A_{\updatee}, A_{\final})$} if $A(I,S)$ can be computed as in Algorithm~\ref{alg:decomp}. We denote the the running times of the sub-algorithms of the decomposition $t_{\setup}$, $t_{\updatee}$ and $t_{\final}$, respectively.
\end{definition}

In Algorithm~\ref{alg:decomp}, $\mathbf{x} = (x_1, x_2, \ldots)$ is a vector of variables that is initialized to some values independent of subset $S$ and determines the algorithm $A$'s final return value. $I'$ is an auxiliary data structure or states that only depend on the representation $I$ and is used in subsequent steps for ease of computation; $I'$ can be simply $I$ if no such preprocessing is necessary. Theorem~\ref{thm:exact} can be generalized as follows:

\begin{theorem}\label{thm:gen1}
Assume a cooperative game $(N,v)$ has an alternative representation $I$ and a deterministic algorithm $A$ that computes $v$. If $A$ has a decomposition $(A_{\setup}, A_{\updatee}, A_{\final})$ such that at most $n(I)$ variables $\mathbf{x}$ are used with each taking at most $m(I)$ possible values as $S$ ranges over all subsets of $N$, then the Shapley value can be computed in $O(t_{\setup} + t_{\updatee} m^n |N|^2 + t_{\final} m^n |N|)$ for each agent. In order-specific cases, for Steps 2-4 of Algorithm~\ref{alg:decomp}, the running time is $O(t_{\setup} + t_{\updatee} m^{2n} |N|^2 + t_{\final} m^{2n} |N|)$. Note that $n$ and $m$ are representation-dependent numbers and the argument $I$ has been omitted.
\end{theorem}

\begin{proof}
Given the alternative representation $I$, we compute the Shapley value of agent $i$. We associate $v(S)$ with the final values, $\mathbf{x}_{S,\textrm{final}}$, of $n(I)$ variables $\mathbf{x}$ in $A(I,S)$, for all $S\subseteq N\setminus \{i\}$. We partition $2^{N\setminus \{i\}}$ by the pair $(|S|, \mathbf{x}_{S, \textrm{final}})$ into at most $m^n |N|$ partitions, omitting the argument $I$ from $n$ and $m$. Let $\mathcal{X}$ be the set of all possible final values of the variables $\mathbf{x}$; note that its cardinality is at most $m^n$. We compute the cardinalites of the partitions using dynamic programming. Let $N' = N\setminus\{i\}$, ordered and relabeled $1,\ldots, |N|-1$, and $i=|N|$. For $0\leq j \leq |N|-1$, $0 \leq s \leq j$, and $\mathbf{v} \in \mathcal{X}$, we define $\hat{c}(j,s,\mathbf{v}) = \# \left\{S\subseteq N'_{1,j}: |S|=s, \mathbf{x}_{S, \textrm{final}} = \mathbf{v}\right\}$. Then, $\hat{c}$ has the recurrence relation \[
\textstyle \hat{c}(j,s,\mathbf{v}) = \hat{c}(j-1, s, \mathbf{v}) + \sum_{\mathbf{u}: \AAA_{\updatee}(I',j, \mathbf{u})= \mathbf{v}} \hat{c}(j-1,s-1,\mathbf{u})
\]
with the base case $\hat{c}(0, 0, \mathbf{s}) = 1$, where $\mathbf{s}$ is the initial states of variables $\mathbf{x}$. Using $A_{\setup}$, we compute $I'$ and the inital values $\mathbf{s}$ in $O(t_{\setup})$. Using the recurrence relation and $\AAA_{\updatee}$, we compute $\hat{c}(j,s, \mathbf{v})$ for all $j$, $s$, and $\mathbf{v}$ in time $O(t_{\updatee} m^n |N|^2)$. Note that for a subset $S\subseteq N\setminus \{i\}$, we can compute agent $i$'s marginal contribution to $S$, i.e., $v(S\cup \{i\}) - v(S)$, in $O(t_{\updatee} + t_{\final})$ from the final values of $\mathbf{x}$ associated with the partition that $S$ belongs to, i.e., $\mathbf{x}_{S, \textrm{final}}$; let $m_i(\mathbf{v})$ be the agent $i$'s marginal contribution to subsets associated with $\mathbf{v}\in \mathcal{X}$. By Lemma~\ref{lemma:char2}, 
\[
\textstyle \phi_i =  \sum_{\mathbf{v} \in \mathcal{X}} \sum_{s=0}^{|N|-1} \hat{c}(|N|-1, s,\mathbf{v}) \frac{s!(|N|-s-1)!}{|N|!} m_i(\mathbf{v}),
\]
and the Shapley value can be calculated in time $O((t_{\updatee} + t_{\final}) m^n |N|)$ using the precomputed values of $\hat{c}$. The overall running time is $O(t_{\setup} + t_{\updatee} m^n |N|^2 + (t_{\updatee} + t_{\final}) m^n |N|)$.

Now assume that the agents have to be processed in a specific order determined by representation $I$. For a given $S$ and its final values $\mathbf{x}_{S, \textrm{final}}$, we cannot compute $\mathbf{x}_{S\cup\{i\},\textrm{final}}$ as $\AAA_{\updatee}(I', i, \mathbf{x}_{S, \textrm{final}})$ and compute agent $i$'s marginal contribution to $S$, because it would violate the order if some agents in $S$ have to be processed after $i$. Instead, we associate $S$ with the final values $\mathbf{x}_{S, \textrm{final}}$ and $\mathbf{x}_{S\cup\{i\}, \textrm{final}}$ and partition $2^{N\setminus\{i\}}$ by the tuple $(|S|, \mathbf{x}_{S,\textrm{final}}, \mathbf{x}_{S\cup\{i\}, \textrm{final}})$ into at most $m^{2n}|N|$ partitions, omitting the argument $I$. Following the same argument as before, we get the running time $O(t_{\setup} + t_{\updatee} m^{2n} |N|^2 + t_{\final} m^{2n} |N|)$.
\ifshowqed
\hfill \qed 
\fi
\end{proof}

The following definition and theorem generalize Theorem~\ref{thm:greedy} and can also be considered a specialization of Theorem~\ref{thm:gen1}. See Appendix~\ref{sec:appendixC} for proof details.

\begin{algorithm}[t]
\caption{Computing $A(S)$ with a per-element decomposition $\{(A_{\setup}^e, A_{\updatee}^e, A_{\final}^e)\}_{e\in M}$} \label{alg:elemdecomp} 
\begin{algorithmic}[1]
\State Initialize $S'=\emptyset$
\For{$e\in M$}
\State $A_{\setup}^e(M,w)$ outputs $I'$, $\mathbf{x}$
\State For $i\in S$: $\mathbf{x} = A_{\updatee}^e(I',i,\mathbf{x})$ 
\State If $A_{\final}(I', \mathbf{x}) = 1$, $S' = S'\cup \{e\}$
\EndFor
\State Return $S'$
\end{algorithmic}
\end{algorithm}

\begin{definition}\label{defn:elemdecomp}
Assume a cooperative game $(N,v)$ has an alternative representation $(M,w,A)$ as described in Lemma~\ref{lemma:char1} such that $v(S) = \sum_{e\in A(S)} w(e)$, for all $S\subseteq N$. Algorithm $A$ has a {\em per-element decomposition} $(A_{\setup}^e, A_{\updatee}^e, A_{\final}^e)$ for all $e\in M$ if $A(S)$ can be computed as in Algorithm~\ref{alg:elemdecomp}. We denote the upper bounds, over all $e\in M$, on running times of the sub-algorithms of the per-element decomposition $t_{\setup}$, $t_{\updatee}$ and $t_{\final}$, respectively.
\end{definition}

\begin{theorem}\label{thm:gen2}
Assume a cooperative game $(N,v)$ has an alternative representation $(M,w,A)$, as given in Lemma~\ref{lemma:char1}. If $A$ has a per-element decomposition $(A_{\setup}^e, A_{\updatee}^e, A_{\final}^e)$ for all $e\in M$ such that at most $n(M,w)$ variables $\mathbf{x}$ are used with each taking at most $m(M,w)$ possible values as $S$ ranges over all subsets of $N$ and $e$ over $M$, the Shapley value can be computed in $O((t_{\setup} + t_{\updatee} m^{n} |N|^2 + t_{\final}m^{n}|N|)|M|)$ for each agent. In order-specific cases, for Step 4 of Algorithm~\ref{alg:elemdecomp}, the running time is $O((t_{\setup} + t_{\updatee} m^{2n}|N|^2 + t_{\final}m^{2n}|N|)|M|)$. Note that $n$ and $m$ are representation-dependent numbers and the argument $(M,w)$ has been omitted.
\end{theorem}

The above definitions apply broadly and suggest the following framework for cooperative games that we term {\em algorithmic representation}; we represent each cooperative game $(N,v)$ in terms of an alternative representation $I$ and an accompanying algorithm $A$ that computes $v$. As we can represent any cooperative game by a table with exponentially many entries for $v$ values and a simple lookup algorithm, the algorithmic representation always exist. The main challenge is to determine an ``efficient'' algorithmic representation for cooperative games in general. The algorithmic representation framework subsumes the concise representation schemes in~\cite{Conitzer,Ieong,Aadithya} as these assume specific structures on the alternative representation $I$. It also captures the notion of classes of cooperative games for we can represent a class of cooperative games by a set of alternative representations corresponding to those games in the class. In this framework, Theorems~\ref{thm:gen1} and \ref{thm:gen2} show that if the algorithms for computing $v$ satisfy the decomposability properties outlined in Definitions~\ref{defn:decomp} and \ref{defn:elemdecomp}, then the Shapley value can be computed efficiently as long as these algorithms are efficient.

Using the generalizations, we can reproduce many previous results on efficient computation of the Shapley value up to a (pseudo) polynomial factor in the running time.\footnote{The slightly slower running times can be attributed to our generalizations' inability to derive closed form expressions on a game-by-game basis; for instance, evaluating the sum $\sum_{i=1}^n i$ in $O(n)$ instead of using the identity $\frac{n(n+1)}{2} = \sum_{i=1}^n i$ in $O(1)$. As generalizations apply in a black-box manner, we argue the loss in running time is reasonable for (pseudo) polynomial time computation.} As concrete examples, we prove several such results (and a new one on the Data Mining Problem in Section~\ref{sec:intro}). We defer proofs to Appendix~\ref{sec:appendixC}:

\begin{corollary}(Weighted Majority Games)\label{cor:wmgames}
Assume a cooperative game $(N,v)$ has a representation given by $|N|+1$ nonnegative integers $q, w_1, \ldots, w_{|N|}$ such that $v(S)$ is 1 if $\sum_{i\in S} w_i \geq q$, or 0 otherwise. Then, the Shapley value can be computed in pseudo-polynomial time $O(q |N|^2)$ for each agent. (Identical to \cite{Matsui})
\end{corollary}

\begin{corollary}(MC-net Representation)\label{cor:mcnet}
Assume a cooperative game $(N,v)$ has a marginal-contribution (MC) net representation with boolean rules $R=\{r_1, \ldots, r_m\}$ with each $r_i$ having value $v_i$ and of the form $(p_1 \wedge \ldots \wedge p_a \wedge \neg n_1 \wedge \ldots \wedge \neg n_b)$ such that $v(S) = \sum_{r_i: S \textrm{satisfies } r_i} v_i$ for all $S\subseteq N$.\footnote{If $r=(1 \wedge 2 \wedge \neg 3)$, then $S=\{1,2\}$ satisfies $r$, but $S=\{1, 3\}$ does not.} Then, the Shapley value can be computed in $O(m |N|^2 (\max_i |r_i|)^2)$ for each agent, where $|r|$ is the number of literals in rule $r$. (Compare to $O(m \max_i |r_i|)$, linear time in the representation size, in \cite{Ieong})
\end{corollary}

\begin{corollary}(Multi-Issue Representation)\label{cor:multiissue}
Assume a cooperative game $(N,v)$ has a multi-issue representation with subsets $C_1, \ldots, C_t \subseteq N$ and characteristic functions $v_i:2^{C_i} \rightarrow \mathbb{R}$ for all $i$ such that $v(S) = \sum_{i=1}^t v_i(S\cap C_i)$ for all $S\subseteq N$. Then, the Shapley value can be computed in $O(t2^{\max_i |C_i|} |N|^2 \max_i |C_i|)$ for each agent. (Compare to $O(t2^{\max_i |C_i|})$ in \cite{Conitzer})
\end{corollary}

\begin{corollary}(Data Mining Problem)\label{cor:topk}
Assume a cooperative game $(N,v)$ has a representation given by $|N|+1$ nonnegative integers $k, w_1, \ldots, w_{|N|}$ such that $v(S) = \max_{S'\subseteq S: |S'| \leq k} w(S')$. Then, the Shapley value can be computed in polynomial time $O(|N|^3)$ for each agent. (This is our own problem.) 
\end{corollary}

%

\section{Further Discussion}\label{sec:discuss}

We have introduced a class of cooperative games called budgeted games and investigated the computational complexity of the Shapley value in the knapsack version, knapsack budgeted games, in particular. We presented exact and approximation algorithms for knapsack budgeted games and a pseudo-polynomial time algorithm for closely related greedy knapsack budgeted games. These algorithms have only polynomial dependence on $|N|$, the number of agents, and are more efficient than the na\"ive exponential time algorithm when $|N|$ is large. Our results extend to knapsack budgeted games with multiple budget constraints. We believe knapsack budgeted games are useful in modeling value division problems in real-life scenarios and our algorithms applicable; for example, when finding a profit distribution solution for a joint venture of, say, 100-plus agents. 

We also provided generalizations and proposed the algorithmic representation framework in which we represent each cooperative game in terms of an alternative representation and an accompanying algorithm that computes the underlying value function. We formalized efficient algorithmic representations and used the generalizations to show that computing the Shapley value in those cooperative games with efficient algorithmic representations can be done efficiently. To demonstrate the generalizations' applicability, we proved old and new results on the efficient computation of the Shapley value. 

We note that further improvement to our algorithmic results might be possible. While the exact algorithm in Theorem~\ref{thm:exact} has polynomial time dependence on $|N|$, it is not a pseudo-polynomial time algorithm and the hardness result in Theorem~\ref{thm:nphard} does not preclude the existence of a polynomial time algorithm for the Shapley value computation in the restricted case of $|N| \gg l_{\bin}$.\footnote{In this case, the $O(|N| l_{\bin})$ dynamic programming time algorithm for the classical knapsack problem algorithm in Footnote~\ref{fn:classicaldp} becomes an $O(|N|^2)$ algorithm, and the classical knapsack problem can be solved in polynomial time.} Similarly, we do not know if the results in Theorems~\ref{thm:approx} and \ref{thm:greedy} are the best possible. We pose these as open problems.

Finally, we believe our techniques can have applications beyond the games considered in this paper and to other economic concepts such as the Banzhaf index. It would be also interesting to investigate the computational complexity of the Shapley value in other kinds of budgeted games.



\paragraph{\bf Acknowledgements.}
We would like to thank Vasilis Gkatzelis for his helpful comments.

\bibliographystyle{abbrv}
\bibliography{mybib}

\newpage
\appendix

\section{Missing Proofs from Section~\ref{sec:shapley}}\label{sec:appendixA}

\begin{proof}(Theorem~\ref{thm:nphard})
We reduce the decision version of the classical knapsack problem, which is NP-complete, to the problem of computing the Shapley value. If there is a polynomial time algorithm for the Shapley value computation, we use it to compute the Shapley value of the agents, $\phi_1, \ldots, \phi_{|N|}$. By the efficiency property, $\phi_1 + \cdots + \phi_{|N|} = v(N)$, which is exactly the solution of the optimization version of the classical knapsack problem. Then, we can solve the decision version. 
\ifshowqed
\hfill \qed 
\fi
\end{proof}

\begin{proof}(Lemma~\ref{lemma:char2})
We start from \eqref{eq:shapley_set}:
\begin{align*}
\phi_i &= \sum_{S\subseteq N\setminus \{i\}} f(|S|) (v(S\cup\{i\}) - v(S)) \\
	&=\sum_{S\subseteq N\setminus \{i\}} \sum_{p=1}^{p_i} \sum_{s=0}^{|N|-1} f(s) m_i(p) \Ind(|S|=s, P_i(S)=p) \\
	&=\sum_{p=1}^{p_i} \sum_{s=0}^{|N|-1} f(s) m_i(p) \left(\sum_{S\subseteq N\setminus \{i\}} \Ind(|S|=s, P_i(S)=p) \right) \\
	&=\sum_{p=1}^{p_i} \sum_{s=0}^{|N|-1} c(i,p,s) \frac{s!(N-s-1)!}{N!} m_i(p),
\end{align*}
where $f(x) = \frac{x!(|N|-x-1)!}{|N|!}$ for a nonnegative integer $x$. Given the Shapley value equation, the running time is straightforward to obtain. 
\ifshowqed
\hfill \qed
\fi
\end{proof}

\begin{proof}(Lemma~\ref{lemma:approx})
We bound the agent $i$'s marginal contribution to $S$ with respect to $v'$: $v'(S\cup\{i\}) - v'(S) \leq v(S\cup \{i\}) - (v(S) - \alpha) = v(S\cup \{i\}) -v(S) + \alpha$, and $v'(S\cup \{i\}) - v'(S) \geq v(S \cup \{i\}) - v(S) - \alpha$. Then, marginal contributions computed with respect to $v$ and $v'$ are within $\alpha$ of each other. Using the Shapley value equation $\phi_i = \frac{1}{|N|!} \sum_{\pi \in \Pi} v(P_\pi^i \cup \{i\}) - v(P_\pi^i)$,
\begin{align*}
|\phi_i - \phi'_i| &= \left|\frac{1}{|N|!} \sum_{\pi \in \Pi} (v(P_\pi^i \cup \{i\}) - v(P_\pi^i)) - \frac{1}{|N|!} \sum_{\pi \in \Pi} (v'(P_\pi^i \cup \{i\}) - v'(P_\pi^i)) \right| \\
	&\leq \frac{1}{|N|!} \sum_{\pi \in \Pi} \left|\left(v(P_\pi^i \cup \{i\}) - v(P_\pi^i)\right) -\left(v'(P_\pi^i \cup \{i\}) - v'(P_\pi^i) \right) \right| \\
	&\leq \frac{1}{|N|!} \sum_{\pi \in \Pi} \alpha \\
	&= \alpha.
\end{align*}
\ifshowqed
\hfill \qed  
\fi
\end{proof}

\section{Missing Proofs from Section~\ref{sec:greedy}}\label{sec:appendixB}

\begin{proof}(Lemma~\ref{lemma:char1})
The proof is nearly identical to that of Lemma~\ref{lemma:char2}. We start from \eqref{eq:shapley_set}:
\begin{align*}
\phi_i &= \sum_{S\subseteq N\setminus \{i\}} f(|S|) (v(S\cup\{i\}) - v(S)) \\
	&= \sum_{S\subseteq N\setminus \{i\}} f(|S|) \left(\sum_{e\in A(S\cup \{i\})} w(e) - \sum_{e\in A(S)} w(e) \right) \\
	&= \sum_{S\subseteq N\setminus \{i\}} \sum_{e\in M} \sum_{s=0}^{|N|-1} f(s)w(e) \left(\Ind(C_s,e\in A(S\cup \{i\})) - \Ind(C_s,e\in A(S))  \right)\\
	&= \sum_{e\in M} \sum_{s=0}^{|N|-1} f(s)w(e) \left(\sum_{S\subseteq N\setminus \{i\}} \Ind(C_s,e\in A(S\cup \{i\})) - \sum_{S\subseteq N\setminus \{i\}} \Ind(C_s,e\in A(S))  \right)\\
	&=\sum_{e\in M} \sum_{s=0}^{|N|-1} \left(c_+(i,e,s) - c_-(i,e,s) \right) \frac{s!(|N|-s-1)!}{|N|!} w(e),
\end{align*}
where $f(x) = \frac{x!(|N|-x-1)!}{|N|!}$ for a nonnegative integer $x$ and $C_s$ is the clause $|S|=s$.

Given the Shapley value equation, the running time is straightforward to obtain.
\ifshowqed
\hfill \qed 
\fi
\end{proof}

\begin{proof}(Theorem~\ref{thm:greedy})
We compute the Shapley value $\phi_i$ for some fixed agent $i$. In what follows, we assume that the agents are sorted and reindexed so that $\frac{w_1}{l_1} \geq \ldots \geq \frac{w_{|N|}}{l_{|N|}}$. If there are multiple agents with the same maximum weight for the $\argmax$ operator, we choose the one with the lowest index.

For ease of exposition, we use $\AAA''$ to denote Step 2 of the greedy heuristic so that $\AAA''(S, b)$ is exactly the set $S'$ in $\AAA'(S,b)$. We also drop the bin size $b$ when it is equal to $l_{\bin}$. Note $v(S) = \sum_{k\in \AAA'(S)} w_k = w(\AAA'(S))$ for all $S\subseteq N$. In order to use Lemma~\ref{lemma:char1}, we choose the alternative representation $(M, w, A)$ where $M=N$, the weight function $w$ is such that $w(e) = w_e, \forall e$, and $A$ is the greedy heuristic $\AAA'$. We consider three cases: $e=i$, $e>i$, and $e<i$.

Case 1) $e=i$: For $c_-(i,s,i)$, we count subsets $S\subseteq N \setminus \{i\}$ such that $i \in \AAA'(S)$. Clearly, $c_-(i,s,i) = 0$. For $c_+(i,s,i)$, we count subsets $S\subseteq N\setminus \{i\}$ such that $i\in \AAA'(S\cup\{i\})$. Note $i\in \AAA'(S\cup \{i\}$) if 1) agent $i$ is included in $S'$ in Step 2 of $\AAA'$ and $S'$ is finally returned in Step 3, or 2) agent $i$ is not included in $S'$, but is selected in Step 1 and finally returned in Step 3. To each subset $S\subseteq N\setminus\{i\}$, we divide the set $S$ into two parts, $S_{1, i-1}$ and $S_{i+1, |N|}$, and associate tuple $(s_1, \hat{l}, \hat{w}_1, a_1, s_2, \hat{w}_2, a_2)$ where $|S_{1, i-1}| = s_1$, $l(\AAA''(S_{1,i-1})) = \hat{l}$, $w(\AAA''(S_{1,i-1})) = \hat{w}_1$, $\argmax_{k\in S_{1,i-1}} w_k = a_1$, $|S_{i+1, |N|}| = s_2$, $w(\AAA''(S_{i+1,|N|})) = \hat{w}_2$, and $\argmax_{k\in S_{i+1,|N|}} w_k = a_2$. Note the tuple has enough information to reconstruct the steps of the greedy heuristic and check aforementioned conditions. Since the total number of possible tuples is bounded, we count the number of sets $S\subseteq N\setminus \{i\}$ for each tuple and sum up to determine $c_+(i,s,i)$. We further decouple steps of $\AAA'$ on $S_{1,i-1}$ and $S_{i+1, |N|}$ by using intermediate bin sizes, the parameter $b$ in the coefficient $\hat{c}_b^2$ below, and consider each part independently using smaller tuples $(s_1, \hat{l}, \hat{w}_1, a_1)$ and $(s_2, \hat{w}_2, a_2)$, respectively. 

We define $\hat{c}^1(j,s_1,\hat{l},\hat{w}_1,a_1) = \#\{S \subseteq N_{1,j}: |S|=s_1, l(\AAA''(S)) = \hat{l}, w(\AAA''(S))=\hat{w}_1, \argmax_{k\in S}w_k = a_1\}$, for $1 \leq j \leq i-1$, and $\hat{c}^2_{b}(j,s_2,\hat{w}_2,a_2) = \#\{S \subseteq N_{i+1, j}: |S| = s_2, w(\AAA''(S,b))=\hat{w}_2, \argmax_{k\in S}w_k = a_2\}$, for $i+1\leq j \leq |N|$ and $0\leq b \leq l_{\bin}$. Then, the number of subsets $S\subseteq N\setminus \{i\}$ with tuple $(s_1, \hat{l}, \hat{w}_1, a_1, s_2, \hat{w}_2, a_2)$ with $i$ is contained in $S'$ in Step 2 of $\AAA'(S)$ is the product of the number of subsets $S_1 \subseteq N_{1,i-1}$ with tuple $(s_1, \hat{l}, \hat{w}_1, a_1)$ and the number of subsets $S_2 \subseteq N_{i+1, |N|}$ with tuple $(s_2, \hat{w}_2, a_2)$, where the effective bin size for $S_2$ is $l_{\bin}-\hat{l} - l_i$. These sets satisfy the first set of conditions for $i\in \AAA'(S\cup\{i\})$ if and only if $\Ind(\hat{l}\leq l_{\bin}-l_i, \hat{w}_1+\hat{w}_2+w_i \geq \max\{w_{a_1}, w_{a_2}\}) = 1$. We reason similarly for the second set of conditions and get $c_+(i,s,i) = \sum \left(\hat{c}^1\cdot \hat{c}^2_{l_{\bin}-\hat{l}-l_i} \cdot \Ind(C_1, C_2)+ \hat{c}^1\cdot \hat{c}^2_{l_{\bin}-\hat{l}} \cdot \Ind(C_3, C_4, C_5) \right)$, where the summation is over $s_1+s_2 = s$, $1\leq a_1 < a_2 \leq |N|$, $0 \leq \hat{l} \leq l_{\bin}$, $0\leq \hat{w}_1, \hat{w}_2 \leq w_{\max}|N|$; $\hat{c}^1$'s have argument $(i-1,s_1, \hat{l}, \hat{w}_1, a_1)$ and $\hat{c}^2$'s have argument $(|N|, s_2, \hat{w}_2, a_2)$; the conditions are $C_1 = (\hat{l}\leq l_{\bin}-l_i)$, $C_2=( \hat{w}_1+\hat{w}_2+w_i \geq \max\{w_{a_1}, w_{a_2}\})$, $C_3 =(\hat{l}>l_{\bin}-l_i)$, $C_4 = (i = \argmax_{k\in\{a_1,a_2,i\}} w_k)$, and $C_5 = (w_i >\hat{w}_1+\hat{w}_2)$. Using dynamic programming, $\hat{c}^1$ can be computed in $O(l_{\bin} w_{\max} |N|^3)$ and $\hat{c}^2_b$ for all $0\leq b\leq l_{\bin}$ can be computed in $O(l_{\bin} w_{\max} |N|^3)$. Hence, each coefficient $c_+(i,s,i)$ can be computed in $O(l_{\bin}{w_{\max}}^2|N|^4)$ using precomputed values of $\hat{c}^1$ and $\hat{c}^2_b$.

The analyses for the other cases are similar. We divide the set $N$ at $i$ and $e$ into three parts and associate each subset satisfying $e\in \AAA'(S\cup \{i\})$ or $e\in \AAA'(S)$ a tuple that summarizes the steps of the greedy heuristic over these parts. Then, we proceed as before and compute each coefficient, $c_+(i,s,e)$ and $c_-(i,s,e)$, in $O({l_{\bin}}^5 {w_{\max}}^5 |N|^{6})$. 

Putting all the cases together, we can compute the necessary coefficients $c(i,e,s)$ in time $O({l_{\bin}}^5 {w_{\max}}^5 |N|^{6})$ per coefficient. By Lemma~\ref{lemma:char1}, we can compute the Shapley value in pseudo-polynomial time of $O({l_{\bin}}^5 {w_{\max}}^5 |N|^{8})$. 
\ifshowqed
\hfill \qed
\fi 
\end{proof}

\section{Missing Proofs from Section~\ref{sec:general}}\label{sec:appendixC}

\begin{proof}(Theorem~\ref{thm:gen2})
Given the alternative representation $(M,w)$, we compute the Shapley value of agent $i$ using Lemma~\ref{lemma:char1}. We show how to compute the quantity $\phi_{i,e} = \sum_{s=0}^{|N|-1} (c_+(i,s,e) - c_-(i,s,e)) \frac{s!(|N|-s-1)!}{|N|!}$ for some arbitrary $e\in M$. We associate $v(S)$ with the final states, $\mathbf{x}_{S, \textrm{final}}$, of $n(M,w)$ variables $\mathbf{x}$ in $A$ for all $S\subseteq N\setminus \{i\}$. We partition $2^{N \setminus \{i\}}$ by the pair $(|S|, \mathbf{x}_{S, \textrm{final}})$ into at most $m^{n}|N|$ partitions, omitting the argument $(M,w)$ from $n$ and $m$.Let $\mathcal{X}$ be the set of all final values of the variables $\mathbf{x}$ in $A$. We compute the cardinalities of the partitions using dynamic programming. Let $N' = N\setminus\{i\}$, ordered and relabeled $1,\ldots, |N|-1$ and $i=|N|$. For $0\leq j \leq |N|-1$, $0 \leq s \leq j$, and $\mathbf{v} \in \mathcal{X}$, we define $\hat{c}(j,s,\mathbf{v}) = \# \left\{S\subseteq N'_{1,j}: |S|=s, \mathbf{x}_{S, \textrm{final}} = \mathbf{v}\right\}$. $\hat{c}$ has the recurrence relation $\hat{c}(j,s,\mathbf{v}) = \hat{c}(j-1, s, \mathbf{v}) + \sum_{\mathbf{u}: \AAA_{\updatee}^e(I',j, \mathbf{u})= \mathbf{v}} \hat{c}(j-1,s-1,\mathbf{u})$ with the base case $\hat{c}(0, 0, \mathbf{s}) = 1$, where $\mathbf{s}$ is the initial states of variables $\mathbf{x}$. Using $A_{\setup}^e$, we compute $I'$ and the inital values $\mathbf{s}$ in $O(t_{\setup})$. Using the recurrence relation and $\AAA_{\updatee}^e$, we compute $\hat{c}(j,s, \mathbf{v})$ for all $j$, $s$, and $\mathbf{v}$ in time $O(t_{\updatee} m^n |N|^2)$. Note that for a subset $S\subseteq N\setminus \{i\}$, we can compute if $e\in A(S)$ and if $e\in A(S\cup\{i\})$ in $O(t_{\final})$ using the information associated with the partition that $S$ belongs to. Then, we can compute the quantity $\phi_{i,e}$ in $O(t_{\final}m^{n}|N|)$ using precomputed values of $\hat{c}$, and the Shapley value $\phi_i = \sum_{e\in M} \phi_{i,e}$ in $O((t_{\setup} + t_{\updatee} m^{n} |N|^2 + t_{\final}m^{n}|N|)|M|)$ overall.

For the order-specific case, we follow the same line of reasoning as in Theorem~\ref{thm:gen1}.
\ifshowqed
\hfill \qed
\fi 
\end{proof}

\begin{proof}(Corollary~\ref{cor:wmgames})
We use the order-agnostic version of Theorem~\ref{thm:gen1} with representation $I=(q, w_1, \ldots, w_{|N|})$. We construct algorithm $A$ with decomposition $(A_{\setup}, A_{\updatee}, A_{\final})$ as follows. As variables $\mathbf{x}$, we use an indicator variable taking values $0, \ldots, q$ to represent that total weight of subset of agents, capped at $q$. In $A_{\setup}$, we initialize $\mathbf{x}$ to 0. In $A_{\updatee}$, for given $i\in S$, we update the indicator variable by adding $w_i$ to its value while capping the total at $q$, In $A_{\final}$, we return $1$ if $\mathbf{x}=q$, or $0$ otherwise. Therefore, $n(I)=1$, $m(I)=q+1$, $t_{\setup} = O(1)$, $t_{\updatee} = O(1)$, and $t_{\final} = O(1)$. This leads to an $O(q|N|^2)$ algorithm for computing the Shapley value for each agent.
\ifshowqed
\hfill \qed
\fi 
\end{proof}

\begin{proof}(Corollary~\ref{cor:mcnet})
We use the order-agnostic version of Theorem~\ref{thm:gen2} with representation $(M,w,A)$ where $M=\{e_r: r\in R\}$, $w(e_r) = v_r$ for each rule $r$, and $A(S) = \{e_r : S \textrm{ satisfies } r\}$. For each $e_r\in M$, we construct a per-element decomposition $(A_{\setup}, A_{\updatee}, A_{\final})$ for $A$, omitting the superscript $e$. Assume the rule $r$ is $(p_1 \wedge \ldots \wedge p_a \wedge \neg n_1 \wedge \ldots \wedge \neg n_b)$ where each literal represents an agent. In $A_{\setup}$, we create a $0/\pm1$ array of length $|N|$ with +1's in the positions corresponding to agents with positive literals in $r$, $-1$'s in the positions corresponding to agents with negative literals, and $0$'s elsewhere. As variables $\mathbf{x}$, we use two counters: $x_1$ that is initialized to $0$ and incremented by 1 for each occurrence of positive literal agent in $S$, and $x_2$ that is initialized to $b$ and decremented by 1 for each occurrence of negative literal agent in $S$. In words, $x_1$ is equal to the number of positive literal agents in $S$ and $b-x_2$ is equal to the number of negative literal agents in $S$. In $A_{\updatee}$, for given $i\in S$, we check the agent's value in the array and update the variables accordingly. In $A_{\final}$, we output 1 if $x_1=a$ and $x_2 = b$, or 0 otherwise. Therefore, $n(M,w) = 2$, $m(M,w) = max_i |r_i|$, $t_{\setup} = O(|N|)$, $t_{\updatee} = O(1)$, and $t_{\final} = O(1)$. This leads to an $O(m |N|^2 (\max_i |r_i|)^2)$ algorithm for computing the Shapley value for each agent.
\ifshowqed
\hfill \qed
\fi 
\end{proof}

\begin{proof}(Corollary~\ref{cor:multiissue})
We use the order-agnostic version of Theorem~\ref{thm:gen2} with representation $(M,w,A)$ where $M=\{e_{iC}: 1\leq i \leq t, C\subseteq C_i\}$, $w(e_{iC}) = v_i(C)$ for all $i$ and $C\subseteq C_i$, and $A(S) = \{e_{1, S\cap C_1}, \ldots, e_{t, S\cap C_t}\}$ for all $S\subseteq N$. Note $|M| = \sum_i 2^{|C_i|}$. We construct per-element decomposition $(A_{\setup}, A_{\updatee}, A_{\final})$ for each $e\in M$ for $A$, omitting superscript $e$. Assume $e_{iC} \in M$. As variable $\mathbf{x}$, we use a counter for the number of elements of $S$, to be given as an input, in $C$. In $A_{\setup}$, we initialize a $0/1$ array of length $|N|$ with 1's in the positions corresponding to elements in $C$ and 0's elsewhere.  In $A_{\updatee}$, for given $i\in S$, we increment $\mathbf{x}$ if $i \in C$. In $A_{\final}$, we check that $\mathbf{x}=|C|$. Therefore, $n(M,w) = 1$, $m(M,w)=\max_i |C_i|$, $t_{\setup} = O(|N|)$, $t_{\updatee} = O(1)$, and $t_{\final} = O(1)$. This leads to an $O(t2^{\max_i |C_i|} |N|^2 \max_i |C_i|)$ algorithm for computing the Shapley value for each agent.
\ifshowqed
\hfill \qed 
\fi
\end{proof}

\begin{proof}(Corollary~\ref{cor:topk})
We use the order-agnostic version of Theorem~\ref{thm:gen2} with representation $(M,w,A)$ where $M=N$, $w(i) = w_i$ for all $i$, and \newline $A(S) = \{\textrm{at most $k$ top agents by weight in $S$}\}$. For simplicity, we assume the weights are all different; otherwise, we break ties consistenetly. For each $e\in M$, we construct a per-element decomposition $(A_{\setup}, A_{\updatee}, A_{\final})$, omitting the superscript $e$. As variables $\mathbf{x}$, we use a counter that counts the number of agents with weight greater than agent $e$'s, capped at $k$. In $A_{\setup}$, we initialize the counter $\mathbf{x}$ to 0. In $A_{\updatee}$, for given $i\in S$, we increment the counter if $w_i>w_e$, capping the counter at $k$. In $A_{\final}$, we output 1 if $\mathbf{x}<k$, or 0 otherwise. Therefore, $n(M,w) = 1$, $m(M,w)=k+1$, $t_{\setup} = O(1)$, $t_{\updatee} = O(1)$, and $t_{\final} = O(1)$. This leads to an $O(k |N|^2)$ algorithm for computing the Shapley value for each agent.

If $k\geq |N|$, then $v(S) = \sum_{i\in S} w_i$ for all $S$ and the Shapley value of agent $i$ is simply $w_i$. So, we can compute the Shapley value in $O(1)$ for each agent after reading the representation. If $k < |N|$, then the above algorithm suffices and its running time reduces to $O(|N|^3)$. In either case, we get a polynomial time algorithm.
\ifshowqed
\hfill \qed
\fi 
\end{proof}

\end{document}